\documentclass[11pt]{article}%
\usepackage{amssymb}
\usepackage{amsfonts}
\usepackage{amsmath}
\usepackage{graphicx}%
\providecommand{\U}[1]{\protect\rule{.1in}{.1in}}
\newtheorem{theorem}{Theorem}
\newtheorem{acknowledgement}[theorem]{Acknowledgement}

\newtheorem{definition}[theorem]{Definition}
\newtheorem{example}[theorem]{Example}

\newtheorem{lemma}[theorem]{Lemma}

\newtheorem{proposition}[theorem]{Proposition}
\newtheorem{remark}[theorem]{Remark}

\newenvironment{proof}[1][Proof]{\noindent\textbf{#1.} }{\ \rule{0.5em}{0.5em}}
\begin{document}

\title{Pointillisme \`{a} la Signac and Construction of a Quantum Fiber Bundle Over
Convex Bodies}
\author{Maurice de Gosson\thanks{Corresponding author: maurice.de.gosson@univie.ac.at}
and Charlyne de Gosson\\University of Vienna\\Faculty of Mathematics (NuHAG)\\Oskar-Morgenstern-Platz 1\\1090 Vienna AUSTRIA}
\maketitle
\tableofcontents

\begin{abstract}
We use the notion of polar duality from convex geometry and the theory of
Lagrangian planes from symplectic geometry to construct a fiber bundle over
ellipsoids that can be viewed as a quantum-mechanical substitute for the
classical symplectic phase space. The total space of this fiber bundle
consists of geometric quantum states, products of convex bodies carried by
Lagrangian planes by their polar duals with respect to a second transversal
Lagrangian plane.. Using the theory of the John ellipsoid we relate these
geometric quantum states to the notion of \textquotedblleft quantum
blobs\textquotedblright\ introduced in previous work; quantum blobs are the
smallest symplectic invariant regions of the phase space compatible with the
uncertainty principle. We show that the set of equivalence classes of
unitarily related geometric quantum states is in a one-to-one correspondence
with the set of all Gaussian wavepackets.

\end{abstract}


\section{Introduction}

\subsection{Pointillisme \`{a} la Signac and phase space pixels}

In two brilliant publications \cite{Butter1,Butter2} Jeremy Butterfield
dismisses what he calls \textit{pointillisme, }that is the view that
mathematical \textit{points} make sense in physics. We totally agree with
Butterfield's views and assume in this paper that the basic elements of
configuration space (\textit{i.e.} physical space, and its multi-dimensional
extensions) are infinitesimal regions with non-zero volume. Indeed, in
practice we can never experimentally determine a point in physical space with
absolute precision; as Gazeau \cite{ga18} humorously notes

\begin{quotation}
\textquotedblleft\textit{Nothing is mathematically exact from the physical
point of view}\textquotedblright.
\end{quotation}

In fact the notion of point-like particle is a mathematical abstraction, which
we can (in principle) approximate with arbitrary accuracy. However, these
regions cannot be made arbitrarily small, because the uncertainty principle
would then lead to violations of special relativity (at least for massive
particles) since in the limit $\Delta x\rightarrow0$ the Heisenberg relation
$\Delta p\Delta x\sim\hbar$ leads to values of $\Delta p$ exceeding the speed
of light. Our view in a sense restores pointillisme as meant by the
neo-impressionist painter Paul Signac, who used small, distinct dots of color
which he applied in patterns to form an image. We will show that this coarse
graining\ of the usual configuration space leads, using an extended version of
the geometric notion of polar duality, to a fiber bundle which can be viewed
as a substitute for a quantum phase space. Admittedly, the term
\textquotedblleft quantum phase space\textquotedblright\ is usually perceived
as a heresy in the physics community: there can't be any phase space in
quantum mechanics since the notion of a well-defined point does not make sense
because of the uncertainty principle. Dirac himself dismissed in 1945 in a
letter to Moyal (in \cite{moyalann}), even the suggestion that quantum
mechanics can be expressed in terms of classical-valued phase space variables.
Of course, as we know, Dirac was wrong, since the Wigner--Moyal--Weyl
formalism, which deals with functions and operators defined on classical phase
space, is one of the most powerful tools for expressing the laws of quantum
mechanics. Still, the concept of \textit{quantum phase space} itself is
ambiguous, to say the least; the aim of this paper is to propose a substitute,
which is a collection of fiber bundles. The simplest of these is the
\textquotedblleft canonical bundle\textquotedblright\
\begin{equation}
\pi_{\mathrm{can}}:\operatorname*{Quant}(n)\longrightarrow\operatorname*{Conv}%
(n)\label{fib}%
\end{equation}
where $\operatorname*{Conv}(n)$ is the set of convex bodies in configuration
space $\mathbb{R}_{x}^{n}$; the fiber over $X\in\operatorname*{Conv}(n)$
consists of the Cartesian products $X\times X^{\hbar}(x_{0})$ where $X^{\hbar
}(p_{0})$ is the polar dual of $X$ centered at $p_{0}\in\mathbb{R}_{p}^{n}$.
For instance%
\[
\pi^{-1}(B_{X}^{n}(x_{0}\sqrt{\hbar}))=\left\{  B_{X}^{n}(x_{0}\sqrt{\hbar
})\times B_{P}^{n}(p_{0},\sqrt{\hbar}):p_{0}\in\mathbb{R}_{p}^{n}\right\}
\]
where $B_{X}^{n}(x_{0},\sqrt{\hbar})$ and $B_{P}^{n}(p_{0},\sqrt{\hbar})$ are
balls with radius $\sqrt{\hbar}$ centered at $x_{0}$ and $p_{0}$; this
reduces, in the limit $\hbar\rightarrow0$, to the products $\{x_{0}%
\}\times\mathbb{R}_{p}^{n}$. We will draw several consequences from these
definitions. In particular we will see that if we restrict the base space of
the fiber bundle (\ref{fib}) to ellipsoids, then we have a continuous action
of the unitary group $U(n,\mathbb{C})$ on $\operatorname*{Quant}(n)$ and that
the homogeneous space $\operatorname*{Quant}(n)/U(n,\mathbb{C})$ can be
identified with the set $\operatorname*{Gauss}(n)$ of all generalized Gaussian
wavepackets on $\mathbb{R}_{x}^{n}$.%

\begin{figure}[ptb]%
\centering
\includegraphics[
natheight=2.302600in,
natwidth=2.950300in,
height=2.3725in,
width=3.0331in
]%
{Paul-Signac-Samois-Study-No.-6.jpg}%
\end{figure}

\subsection{Description of the method: heuristics}

The aim of the present paper is to study, for an arbitrary number $n$ of
degrees of freedom, the properties of such \textquotedblleft quantum
state\textquotedblright\ and to relate them to the theory of Gaussian
wavepackets; our study will unveil unexpected and beautiful geometric
properties of quantum mechanics.

\subsection{Toolbox and terminology}

We introduced in \cite{gopolar} the geometric notion of Lagrangian polar
duality in connection with the uncertainty principle of quantum mechanics; in
a recent paper \cite{bullsci} we have detailed this results and given a
rigorous mathematical study of this notion. As pointed out in \cite{gopolar}
the underlying idea is that a quantum system localized in the position
representation in a set $X$ cannot be localized in the momentum representation
in a set smaller than its polar dual $X^{\hbar}$; this is a geometric form of
the uncertainty principle, independent of the notion of variance or
covariance. Let us explain this a little bit more in detail. We live in a
three-dimensional world where the state of a classical particle is described
by its position vector $(x,y,z)$ and by the vector of conjugate momenta
$(p_{x},p_{y},p_{z})$, both at a given time $t$. This extends to many particle
systems by introducing the generalized position and momentum vectors
$x=(x_{1},...,x_{n})$ and $p=(p_{1},...,p_{n})$, and the phase space of that
system is by definition the space $\mathbb{R}_{x}^{n}\times\mathbb{R}_{p}%
^{n}\equiv\mathbb{R}^{2n}$ of all $(x,p)$. This way of writing things
explicitly singles out the two subspaces $\ell_{X}=\mathbb{R}_{x}^{n}\times0$
and $\ell_{P}=0\times\mathbb{R}_{p}^{n}$; however, as is already clear in
classical (Hamiltonian) mechanics this \textquotedblleft
canonical\textquotedblright\ choice of frame $(\ell_{X},\ell_{P})$ has no
reason to be privileged, and one can choose any other coordinate spaces to
work with as long as these are obtained by symplectic transformations from the
frame $(\ell_{X},\ell_{P})$. Such transformations will not take $\ell_{X}$ and
$\ell_{P}$ to arbitrary $n$-dimensional linear subspaces of $\mathbb{R}^{2n}$,
but rather to \emph{Lagrangian planes} which have the property that the
canonical symplectic form on $\mathbb{R}^{2n}$ vanishes identically on them.
These subspaces play a central role in classical mechanics (they are the
tangent spaces of the invariant tori of the integrable Hamiltonian systems
\cite{Arnold}). Consider now a convex compact set $X_{\ell}$ with non-empty
interior (for instance an ellipsoid) carried by a Lagrangian plane $\ell$. if,
for instance, $\ell=\ell_{X}$ this convex body $X_{\ell}$ can be physically
interpreted as a cloud of points in configuration space corresponding to a
sequence of measurements. Assuming, for simplicity, that $X_{\ell}$ is
centered at the origin, we next choose a second arbitrary Lagrangian plane
$\ell^{\prime}$ transversal to $\ell$ and define the polar dual $X_{\ell
^{\prime}}^{\hbar}$of $X_{\ell}$ with respect to $\ell^{\prime}$ as being the
set of all phase space points $z^{\prime}=(x,p^{\prime})$ such that
$\omega(z,z^{\prime})\leq\hbar$ for every $z=(x,p)$ in $X_{\ell}$. An
elementary argument shows that $X_{\ell^{\prime}}^{\hbar}$ is also a convex
set (and in particular an ellipsoid if $X_{\ell}$ is). We will call the subset
$X_{\ell}\times X_{\ell^{\prime}}^{\hbar}$ of $\mathbb{R}^{2n}$ a \emph{pure
quantum state}. Admittedly, this definition of a quantum state is rather
abstract. The reason will become clear to the reader in the course of this
article, but there is a rather immediate (although hidden) motivation. It
turns out that the Cartesian product $X_{\ell}\times X_{\ell^{\prime}}^{\hbar
}$ is always a convex set (because $X_{\ell}$ and $X_{\ell^{\prime}}^{\hbar}$
are convex). As such it contains a unique maximum volume ellipsoid $\Omega$
(the \textquotedblleft John ellipsoid\textquotedblright), and this ellipsoid
is what we have called elsewhere \cite{blob,goluPR} a \emph{quantum blob},
that is the image of a phase space ball with radius $\sqrt{\hbar}$ by a
symplectic transformation. As we have shown in \cite{go09,goluPR} these
quantum blobs represent the smallest phase space units compatible with the
uncertainty (or indeterminacy) principle of quantum mechanics. In particular,
a quantum blob can always (via the theory of the Wigner transform) be viewed
as the covariance ellipsoid of a generalized Gaussian state.

Here is a basic example. Suppose that the configuration space is the $x$ axis,
in which case the classical phase space is just the $x,p$ plane. The pseudo
quantum phase space consists of parallelograms $X_{\ell}\times X_{\ell}%
^{\hbar}$ where $\ell$ and $\ell^{\prime}$ are two lines in the the $x,p$
plane, $X_{\ell}$ is an interval in $\ell$ and $X_{\ell}^{\hbar}$ is the polar
dual of $X_{\ell}$ with respect to $\ell^{\prime}$. The latter is the set of
points $z^{\prime}$ on $\ell^{\prime}$ such that
\[
\omega(z^{\prime},z)=-%
\begin{vmatrix}
x^{\prime} & x\\
p^{\prime} & p
\end{vmatrix}
\leq\hbar
\]
for all $z=(x,p)$ on $\ell$. If $\ell$ is the $x$-axis and $\ell^{\prime}$ the
$p$-axis this condition becomes $p^{\prime}x\leq\hbar$ so $X_{\ell}^{\hbar}$
is the usual polar dual from convex geometry \cite{gopolar}. Choosing
$X_{\ell_{X}}=[-\sqrt{\hbar/a},\sqrt{\hbar/a}]$ for some $a>0$ we have
$X_{\ell_{P}}^{\hbar}=[-\sqrt{a\hbar},\sqrt{a\hbar}]$ so that $X_{\ell_{X}%
}\times X_{\ell_{P}}^{\hbar}$ is a parallelogram with area $4\hbar$ centered
at the origin. Now, the largest ellipse contained in that parallelogram is the
one with axes $X_{\ell_{X}}$ and $X_{\ell_{P}}^{\hbar}$ and thus has area
$\pi\hbar$. To such an ellipse corresponds (via the theory of the Wigner
transform) a unique (normalized) Gaussian wavepacket, namely%
\[
\psi(x)=\left(  \tfrac{a}{\pi\hbar}\right)  ^{1/4}e^{-ax^{2}/\hbar}%
\]
which is a minimum uncertainty wavepacket. To our \textquotedblleft quantum
state\textquotedblright\ $X_{\ell_{X}}\times X_{\ell_{P}}^{\hbar}$ thus
corresponds a basic object from quantum mechanics (a Gaussian wavepacket), but
is a more general object than just this wavepacket.

\subsubsection*{Notation}

The configuration space of a system with $n$ degrees of freedom will in
general be written $\ell_{X}=\mathbb{R}_{x}^{n}$, and its dual (the momentum
space) $\ell_{P}=\mathbb{R}_{p}^{n}$. The position variables will be written
$x=(x_{1},...,x_{n})$ and the momentum variables $p=(p_{1},...,p_{n})$. The
classical phase space $\mathbb{R}_{x}^{n}\times\mathbb{R}_{p}^{n}$ is
identified with $\mathbb{R}^{2n}$ equipped with the inner product $p\cdot
x=p_{1}x_{1}+\cdot\cdot\cdot+p_{n}x_{n}$ and with the standard symplectic form
$\omega$ defined by $\omega(z,z^{\prime})=p\cdot x^{\prime}-p^{\prime}\cdot x$
if $z=(x,p)$, $z^{\prime}=(x^{\prime},p^{\prime})$.

\section{Some Symplectic Geometry}

\subsection{The symplectic group $\operatorname*{Sp}(n)$}

The standard symplectic form $\omega$ on $\mathbb{R}_{z}^{2n}\equiv
\mathbb{R}_{x}^{n}\times\mathbb{R}_{p}^{n}$ can be written in matrix form as
\[
\omega(z,z^{\prime})=Jz\cdot z^{\prime}=(z^{\prime})^{T}Jz
\]
where $J$ is the standard symplectic matrix:%
\[
J=%
\begin{pmatrix}
0_{n\times n} & I_{n\times n}\\
-I_{n\times n} & 0_{n\times n}%
\end{pmatrix}
.
\]
The associated symplectic group $\operatorname*{Sp}(n)$ consists of all linear
automorphisms $S$ of $\mathbb{R}_{z}^{2n}$ preserving the symplectic form:
$\omega(Sz,Sz^{\prime})=\omega(z,z^{\prime})$ for all vectors $z,z^{\prime}$.
The symplectic automorphisms will be identified with their matrices in the
canonical basis; with this convention $S\in\operatorname*{Sp}(n)$ if and only
it satisfies one of the equivalent identities $S^{T}JS=J$ or $SJS^{T}=J$.
These relations imply \cite{Birk} that a real $2n\times2n$ matrix written in
the block form
\begin{equation}
S=%
\begin{pmatrix}
A & B\\
C & D
\end{pmatrix}
\label{ABCD}%
\end{equation}
is symplectic if and only if the $n\times n$ blocks $A,B,C,D$ satisfy the sets
of equivalent conditions
\begin{align}
A^{T}C\text{, }B^{T}D\text{ \ \textit{symmetric, and} }A^{T}D-C^{T}B  &
=I_{n\times n}\label{cond1}\\
AB^{T}\text{, }CD^{T}\text{ \ \textit{symmetric, and} }AD^{T}-BC^{T}  &
=I_{n\times\times n}. \label{cond2}%
\end{align}
It follows that the inverse of $S\in\operatorname*{Sp}(n)$ has the simple form%
\begin{equation}
S^{-1}=%
\begin{pmatrix}
D^{T} & -B^{T}\\
-C^{T} & A^{T}%
\end{pmatrix}
. \label{inverse}%
\end{equation}

The affine (or inhomogeneous) symplectic group is the semi-direct product
\begin{equation}
\operatorname*{ISp}(n)=\operatorname*{Sp}(n)\ltimes\mathbb{R}^{2n};
\label{ispn}%
\end{equation}
it consists of all products $ST(z_{0})=T(Sz_{0})S$ where $S\in
\operatorname*{Sp}(n)$ and $T(z_{0})$ is the translation operator
$z\longmapsto z+z_{0}$ in $\mathbb{R}^{2n}$.

Recall \cite{Birk} that the metaplectic group $\operatorname*{Mp}(n)$ is the
unitary representation on $L^{2}(\mathbb{R}_{x}^{n})$ of the double cover of
the symplectic group $\operatorname*{Sp}(n)$. It is generated by the unitary
operators $\widehat{J}$, $\widehat{V}_{P}$, and $\widehat{M}_{L.m}$ defined in
the table below, where we denote $\pi^{\operatorname*{Mp}}$ the projection
$\operatorname*{Mp}(n)\longrightarrow\operatorname*{Sp}(n).$

\begin{center}%
\begin{tabular}
[c]{|l|l|l|}\hline
$\widehat{J}\psi(x)=\left(  \tfrac{1}{2\pi i\hbar}\right)  ^{n/2}\int
e^{-\frac{1}{\hbar}x\cdot x^{\prime}}\psi(x^{\prime})d^{n}x^{\prime}$ &
$\overset{\pi^{\operatorname*{Mp}}}{\longrightarrow}$ & $J=%
\begin{pmatrix}
0 & I\\
-I & 0
\end{pmatrix}
$\\\hline
$\widehat{V}_{P}\psi(x)=e^{-\frac{i}{2\hbar}Px\cdot x}\psi(x)$ &
$\overset{\pi^{\operatorname*{Mp}}}{\longrightarrow}$ & $V_{P}=%
\begin{pmatrix}
I & 0\\
-P & I
\end{pmatrix}
$\\\hline
$\widehat{M}_{L.m}\psi(x)=i^{m}\sqrt{|\det L|}\psi(Lx)$ & $\overset{\pi
^{\operatorname*{Mp}}}{\longrightarrow}$ & $M_{L}=%
\begin{pmatrix}
L^{-1} & 0\\
0 & L^{T}%
\end{pmatrix}
.$\\\hline
\end{tabular}

\end{center}

In the last line of this table the integer $m$ is defined modulo 4 and
corresponds to a choice of the argument of the determinant $\det L$,
reflecting the fact that $\operatorname*{Mp}(n)$ is a double covering of
$\operatorname*{Sp}(n)$. For a complete study of $\operatorname*{Mp}(n)$ and
its properties we refer to \cite{Birk}. The non-homogeneous analogue of
$\operatorname*{Mp}(n)$ is denoted $\operatorname*{IMp}(n)$; it consists of
all operators $\widehat{S}\widehat{T}(z_{0})=\widehat{T}(Sz_{0})\widehat{S}$
where $\widehat{S}\in\operatorname*{Mp}(n)$, $z_{0}\in\mathbb{R}^{2n}$, and
$\widehat{T}(z_{0})$ is the Heisenberg displacement operator:%
\[
\widehat{T}(x_{0},p_{0})\psi(x)=e^{\frac{i}{\hbar}(p_{0}\cdot x-\frac{1}%
{2}p_{0}\cdot x_{0})}\psi(x-x_{0}).
\]
The natural projection $\operatorname*{IMp}(n)\longrightarrow
\operatorname*{ISp}(n)$ is defined by $\widehat{S}\widehat{T}(z_{0}%
)\longmapsto ST(z_{0})$.

\subsection{Lagrangian planes and frames}

When $n=1$ the symplectic form is, up to the sign, the determinant function:
$\omega(z,z^{\prime})=-\det(z,z^{\prime})$. It follows that $\omega
(z,z^{\prime})=0$ when $z$ and $z^{\prime}$ are colinear: the symplectic form
vanishes along all lines through the origin. The notion of Lagrangian plane
generalizes this property to arbitrary dimension $n$: a linear subspace $\ell$
of $\mathbb{R}^{2n}$ equipped with its symplectic form $\omega$ is called a
\emph{Lagrangian plane} if $\dim\ell=n$ and $\omega(z,z^{\prime})=0$ for all
$z,z^{\prime}\in\ell$.

The most typical (but not most general) example of Lagrangian planes is given
by the \textquotedblleft coordinate Lagrangian planes\textquotedblright. They
are obtained by picking out in the $2n$-vector $z=(x_{1},...,x_{n}%
;p_{1},...,p_{n})$ exactly $n$ non-conjugate coordinates. For instance the set
of all $(x_{1},...,x_{k},p_{k+1},...,p_{n})$ for $k<n$ are the coordinates of
a Lagrangian plane in $\mathbb{R}^{2n}$.

The subspaces consisting of all $z=(x,p)$ such that $p=Ax$ for some symmetric
matrix $A$ is a Lagrangian plane: it has dimension $n$ and
\[
\omega(x,Ax;x^{\prime},Ax^{\prime})=Ax\cdot x^{\prime}-Ax^{\prime}\cdot x=0
\]
since $A$ is symmetric. More generally, a subspace $\ell$ of $\mathbb{R}^{2n}$
is a Lagrangian plane if and only we have
\[
(x,p)\in\ell\text{ if and only }Ax+Bp=0.
\]
where $A$ and $B$ are real $n\times n$ matrices satisfying one of the
following sets of equivalent conditions
\begin{align*}
A^{T}B  &  =B^{T}A\text{ \textit{and} }A^{T}A+B^{T}B=I_{n\times n}\\
AB^{T}  &  =BA^{T}\text{\textit{ and} }AA^{T}+BB^{T}=I_{n\times n}.
\end{align*}

The set of all Lagrangian planes in the symplectic space $(\mathbb{R}%
^{2n},\omega)$ is called the \textit{Lagrangian Grassmannian} and is denoted
by $\operatorname*{Lag}(n)$.

\begin{remark}
There is an alternative way of interpreting Lagrangian planes as subspaces of
$\mathbb{C}^{2n}$ on which the inner product $(z,z^{\prime})\longmapsto
z\cdot(z^{\prime})^{\ast}$is real. In fact, the symplectic product
$\omega(z,z^{\prime})$ can be written as $\omega(z,z^{\prime}%
)=\operatorname{Im}(z\cdot(z^{\prime})^{\ast})$ when $z=(x,p)$ and $z^{\prime
}=(x^{\prime},p^{\prime})$ are identified with the complex vectors $x+ip$ and
$x^{\prime}+ip^{\prime}$ in $\mathbb{C}^{n}$. Lagrangian planes then
correspond to the $n$-dimensional subspaces for which $z\cdot(z^{\prime
})^{\ast}$ is a real number.
\end{remark}

In the phase plane $\mathbb{R}^{2}$ every line through the origin can be taken
to any other such line using a rotation. There is a similar property in
arbitrary dimension $n$. A symplectic automorphism $U$ is called a
\textit{symplectic} \textit{rotation} if $U\in\operatorname*{Sp}(n)\cap
O(2n,\mathbb{R})$ where $O(2n,\mathbb{R})$ is the usual orthogonal group. In
the case $n=1$ this is just the usual rotation group $SO(2n,\mathbb{R})$. We
denote by $U(n)$ the group of all symplectic rotations; one shows \cite{Birk}
that $U(n)$ is the image in $\operatorname*{Sp}(n)$ of the complex unitary
group $U(n,\mathbb{C})$ by the embedding%
\[
\iota:A+iB\longmapsto%
\begin{pmatrix}
A & B\\
-B & A
\end{pmatrix}
.
\]
A matrix $%
\begin{pmatrix}
A & B\\
-B & A
\end{pmatrix}
$ is thus a symplectic rotation if and only if the blocks $A$ and $B$ satisfy
the conditions%
\begin{align}
A^{T}B  &  =B^{T}A\text{ \ and }A^{T}A+B^{T}B=I\label{uni1}\\
AB^{T}  &  =BA^{T}\text{ \ and }AA^{T}+BB^{T}=I \label{uni2}%
\end{align}
in view of (\ref{cond1}), (\ref{cond2}).

Let $\ell$ be a Lagrangian plane in $(\mathbb{R}^{2n},\omega)$: $\ell
\in\operatorname*{Lag}(n)$. For every symplectic transformation $S\in
\operatorname*{Sp}(n)$ the image $S\ell$ is also a Lagrangian plane: we
clearly have $\dim S\ell=n$ and $\omega(Sz,Sz^{\prime})=\omega(z,z^{\prime
})=0$ for all $z,z^{\prime}\in\ell$. We thus have a natural group action
\begin{equation}
\operatorname*{Sp}(n)\times\operatorname*{Lag}(n)\longrightarrow
\operatorname*{Lag}(n) \label{splag}%
\end{equation}
which induces, by restriction, an action%
\begin{equation}
U(n)\times\operatorname*{Lag}(n)\longrightarrow\operatorname*{Lag}(n).
\label{ulag}%
\end{equation}
An essential property is the transitivity of these actions.

\begin{proposition}
\label{PropU}The subgroup $U(n)$ of $\operatorname*{Sp}(n)$ (and hence
$\operatorname*{Sp}(n)$ itself) acts transitively on the Lagrangian
Grassmannian $\operatorname*{Lag}(n)$: for any pair $(\ell,\ell^{\prime})$ of
Lagrangian planes in $(\mathbb{R}^{2n},\omega)$ there exists $U\in U(n)$ such
that $\ell^{\prime}=U\ell$. In particular every $\ell\in\operatorname*{Lag}%
(n)$ can be obtained from $\ell_{X}$ (or from $\ell_{P}$) using a symplectic rotation.
\end{proposition}

\begin{proof}
This is proven as follows \cite{Birk}: let $\mathcal{B}=\{e_{1},...,e_{n}\}$
and $\mathcal{B}^{\prime}=\{e_{1}^{\prime},...,e_{n}^{\prime}\}$ be
orthonormal bases of $\ell$ and $\ell^{\prime}$, respectively. Then
$\mathcal{B}\cup J\mathcal{B}$ and $\mathcal{B}^{\prime}\cup J\mathcal{B}%
^{\prime}$ are bases of $\mathbb{R}^{2n}$ which are both orthogonal and
symplectic. Let $U$ be a linear mapping taking $\mathcal{B}\cup J\mathcal{B}$
to $\mathcal{B}^{\prime}\cup J\mathcal{B}^{\prime}$; we then have
$\ell^{\prime}=U\ell$ and $U\in\operatorname*{Sp}(n)\cap O(2n,\mathbb{R})$.
\end{proof}

The action (\ref{ulag}) allows to endow $\operatorname*{Lag}(n)$ with a
topology, using the theory of homogeneous spaces. In fact, the subgroup $O(n)$
of $U(n)$ consisting of all symplectic matrices
\[
R=%
\begin{pmatrix}
A & 0\\
0 & A
\end{pmatrix}
\text{ \ , \ }A\in O(n,\mathbb{R})
\]
stabilizes $\ell_{P}$ (that is, $R\ell_{P}=\ell_{P}$) hence there is a natural
bijection
\[
U(n)/O(n)\equiv U(n,\mathbb{C})/O(n,\mathbb{R})\longrightarrow
\operatorname*{Lag}(n)
\]
which allows to identify topologically the coset space $U(n)/O(n)$ with the
Lagrangian Grassmannian (see \cite{Birk} for technical details).

Let $(\ell,\ell^{\prime})$ be a pair of Lagrangian planes in $(\mathbb{R}%
^{2n},\omega)$ such that $\ell\cap\ell^{\prime}=0$. Since the dimensions of
$\ell$ and $\ell^{\prime}$ are $n$ this is equivalent to $\ell\oplus
\ell^{\prime}=\mathbb{R}^{2n}$. We will call $(\ell,\ell^{\prime})$ a
\emph{Lagrangian frame}. We will use the notation
\begin{equation}
\ell_{X}=\mathbb{R}_{x}^{n}\times0\text{ \ and \ }\ell_{P}=0\times
\mathbb{R}_{p}^{n}\label{canonframe}%
\end{equation}
and call the spaces $\ell_{X}$ and $\ell_{P}$ the \textit{position} and
\textit{momentum planes}; Clearly $(\ell_{X},\ell_{P})$ is a Lagrangian frame
(we will call it the \textquotedblleft canonical frame\textquotedblright). We
will denote the space of all Lagrangian frames $\operatorname*{Lag}%
\nolimits_{0}^{2}(n)$. Thus:%
\begin{equation}
\operatorname*{Lag}\nolimits_{0}^{2}(n)=\{(\ell,\ell^{\prime})\in
\operatorname*{Lag}\nolimits^{2}(n):\ell\cap\ell^{\prime}=\{0\}\}\label{lagon}%
\end{equation}
where $\operatorname*{Lag}\nolimits^{2}(n)$ denotes the Cartesian product
$\operatorname*{Lag}(n)\times\operatorname*{Lag}(n)$.

A crucial property is that the symplectic group $\operatorname*{Sp}(n)$ acts
transitively on the set of all Lagrangian frames \cite{Birk}. Because of the
importance of this result we prove it here:

\begin{proposition}
\label{Prop1}The group $\operatorname*{Sp}(n)$ acts transitively on the set of
all Lagrangian frames: if $(\ell_{1},\ell_{1}^{\prime})$ and $(\ell_{2}%
,\ell_{2}^{\prime})$ are in $\operatorname*{Lag}\nolimits_{0}^{2}(n)$ then
there exits $S\in\operatorname*{Sp}(n)$ such that $(\ell_{2},\ell_{2}^{\prime
})=(S\ell_{1},S\ell_{1}^{\prime})$.
\end{proposition}

\begin{proof}
Choose a basis $\mathcal{B=}\{e_{11},...,e_{1n}\}$ of $\ell_{1}$ and a basis
$\mathcal{B}^{\prime}=\{f_{11},...,f_{1n}\}$ of $\ell_{1}^{\prime}$ such that
$\{e_{1i},f_{1j}\}_{1\leq i,j\leq n}$ is a symplectic basis of $(\mathbb{R}%
_{z}^{2n},\omega)$ (\textit{i.e.} $\omega(e_{i1},e_{j1})=\omega(f_{i1}%
,f_{j1})=0$ and $\omega(f_{i1},e_{j1})=\delta_{ij}$ for all $i,j=1,...,n$).
Similarly choose bases of $\ell_{2}$ and $\ell_{2}^{\prime}$ whose union
$\{e_{2i},f_{2j}\}_{1\leq i,j\leq n}$ is also a symplectic basis. Define a
linear mapping $S:\mathbb{R}^{2n}\longrightarrow\mathbb{R}^{2n}$ by
$S(e_{1i})=e_{2i}$ and $S(f_{1i})=f_{2i}$ for $1\leq i\leq n$. We have $S\in$
$\operatorname*{Sp}(n)$ and $(\ell_{2},\ell_{2}^{\prime})=(S\ell_{1},S\ell
_{1}^{\prime})$.
\end{proof}

Notice that we cannot replace $\operatorname*{Sp}(n)$ with $U(n)$ in the
result above. For instance, in the case $n=1$ no rotation will take an
arbitrary pair of transverse of lines to another arbitrary pair of transverse
lines if they do not form equal angles ($U(1)=SO(2,\mathbb{R})$ preserves
angles, while $\operatorname*{Sp}(1)$ does not).

\begin{remark}
It follows from Proposition \ref{Prop1} that every Lagrangian frame in
$(\mathbb{R}^{2n},\omega)$ can be obtained from the canonical frame $(\ell
_{X},\ell_{P})$ using a symplectic transformation.
\end{remark}

The following property is useful when considering phase space shifts of the origin:

\begin{lemma}
\label{Rem1}Every phase space point $z_{0}\in\mathbb{R}^{2n}$ belongs to at
least one Lagrangian plane.
\end{lemma}

\begin{proof}
The case $z_{0}=0$ being trivial we assume $z_{0}\neq0$. Let $e_{1}$ be a
normalized vector such that $z_{0}=\lambda e_{1}$ and choose vectors
$e_{2},...,e_{n}$ and $f_{2},...,f_{n}$ such that $\{e_{1},...,e_{n}%
\}\cup\{f_{1},...,f_{n}\}$ is a symplectic basis of $\mathbb{R}^{2n}$ (this is
a symplectic variant of the Gram--Schmidt orthonormalization process, see
\cite{Birk} for an explicit construction). The subspace spanned by the set of
vectors $\{e_{1},...,e_{n}\}$ is Lagrangian and contains $z_{0}$.
\end{proof}

\subsection{Lagrangian ellipsoids}

Let us identify the position space ellipsoid
\[
X=\{x\in\mathbb{R}_{x}^{n}:Ax\cdot x\leq\hbar\}
\]
with the phase space subset%
\[
X=\{z=(x,0):(A\oplus0)z\cdot z\leq\hbar\}
\]
where, by definition,
\[
A\oplus0=%
\begin{pmatrix}
A & 0_{n\times n}\\
0_{n\times n} & 0_{n\times n}%
\end{pmatrix}
.
\]
The image of $X$ by $S\in\operatorname*{Sp}(n)$ (or by any phase space
automorphism) is then%
\begin{equation}
S(X)=\{z:((S^{T})^{-1}(A\oplus0)S^{-1})z\cdot z\leq\hbar\}. \label{sxa}%
\end{equation}

Let us call \textquotedblleft quantum blob\textquotedblright\ \cite{blob} the
image of the phase space ball $B^{2n}(z_{0},\sqrt{\hbar})$ by a symplectic
transformation. The following property shows that every ellipsoid carried by a
Lagrangian plane $\ell$ is the intersection $\ell\cap Q$ of that subspace with
a quantum blob (or any other phase space ball, for that):

\begin{proposition}
Let $X_{\ell}$ be an $n$-dimensional ellipsoid centered at $z_{0}\in\ell$ and
carried by the Lagrangian plane $\ell\in\operatorname*{Lag}(n)$. There exists
$S\in\operatorname*{Sp}(n)$ such that $X_{\ell}=S(B^{2n}(S^{-1}z_{0}%
,\sqrt{\hbar}))\cap\ell$.
\end{proposition}

\begin{proof}
It is sufficient to assume $z_{0}=0$. We first consider the case $\ell
=\ell_{X}$, then $X_{\ell_{X}}=\{x:Ax\cdot x\leq\hbar\}$ where $A$ is a
symmetric positive definite matrix. Clearly, $X_{\ell_{X}}$ is the
intersection of the phase space ellipsoid
\[
\Omega=\{(x,p):Ax\cdot x+A^{-1}p\cdot p\leq\hslash\}
\]
with $\ell_{X}$, and $\Omega$ is indeed a quantum blob since $\Omega
=S(B^{2n}(\sqrt{\hbar}))$ with
\begin{equation}
S=%
\begin{pmatrix}
A & 0\\
0 & A^{-1}%
\end{pmatrix}
\in\operatorname*{Sp}(n). \label{MSymp}%
\end{equation}
Suppose now $\ell$ is an arbitrary Lagrangian plane. In view of Proposition
\ref{PropU} there exists a symplectic rotation $R\in U(n)$ such that
$\ell=R\ell_{X}$. The set $X_{\ell_{X}}=R^{-1}(X_{\ell})$ is an ellipsoid in
$\ell_{X}$ centered at $z_{0}=0$ and hence $X_{\ell_{X}}=Q\cap\ell_{X}$ for
some quantum blob $Q$, and\ $X_{\ell}=R(X_{\ell_{X}})=(RQ)\cap\ell$ which
concludes the proof since $R(Q)$ is also a quantum blob.
\end{proof}

\begin{remark}
\label{Rem2}The quantum blob described in the result above is not unique. For
instance there exist infinitely many quantum blobs $Q=S(B^{2n}(\sqrt{\hbar}))$
such that $X_{\ell_{X}}=Q\cap\ell_{X}$.
\end{remark}

\section{Lagrangian Polar Duality and Quantum States}

\subsection{Polar duality: review}

We begin by briefly recalling the usual notion of polar duality from convex
geometry (we are following our presentation in \cite{gopolar}); for the
notions of convex geometry we use see for instance \cite{ABMB,rocka,Vershynin}%
). Let $X$ be a convex body in configuration space $\mathbb{R}_{x}^{n}$ (a
convex body is a compact convex set with non-empty interior). We assume in
addition that $X$ contains $0$ in its interior. This is the case if, for
instance, $X$ is symmetric: $X=-X$. The \emph{polar dual} of $X$ is the
subset
\begin{equation}
X^{\hbar}=\{p\in\mathbb{R}_{x}^{n}:\sup\nolimits_{x\in X}(p\cdot x)\leq\hbar\}
\label{omo}%
\end{equation}
of the dual space $\mathbb{R}_{p}^{n}\equiv(\mathbb{R}_{x}^{n})^{\ast}$.
Notice that it trivially follows from the definition that $X^{\hbar}$ is
convex and contains $0$ in its interior. In the mathematical literature one
usually chooses $\hbar=1$, in which case one writes $X^{o}$ for the polar
dual; we have $X^{\hbar}=\hbar X^{o}$. The following properties are straightforward:

\begin{center}%
\begin{tabular}
[c]{|l|l|l|}\hline
\textit{Reflexivity (bipolarity)}: & $(X^{\hbar})^{\hbar}=X$ & P1\\\hline
\textit{Antimonotonicity: } & $X\subset Y\Longrightarrow Y^{\hbar}\subset
X^{\hbar}$ & P2\\\hline
\textit{Scaling property} & $A\in GL(n,\mathbb{R})\Longrightarrow(AX)^{\hbar
}=(A^{T})^{-1}X^{\hbar}$. & P3\\\hline
\end{tabular}

\end{center}

In \cite{gopolar} we proved the following elementary properties of polar duality:

\textit{(i)} Let $B_{X}^{n}(R)$ (\textit{resp}. $B_{P}^{n}(R)$) be the ball
$\{x:|x|\leq R\}$ in $\mathbb{R}_{x}^{n}$ (\textit{resp}. $\{p:|p|\leq R\}$ in
$\mathbb{R}_{p}^{n}$). Then
\begin{equation}
B_{X}^{n}(R)^{\hbar}=B_{P}^{n}(\hbar/R)~. \label{BhR}%
\end{equation}
In particular
\begin{equation}
B_{X}^{n}(\sqrt{\hbar})^{\hbar}=B_{P}^{n}(\sqrt{\hbar}). \label{bhh}%
\end{equation}

\textit{(ii)} Let $A$ be a real invertible and symmetric $n\times n$ matrix
and $R>0$. The polar dual of the ellipsoid defined by $Ax\cdot x\leq R^{2}$ is
given by
\begin{equation}
\{x:Ax\cdot x\leq R^{2}\}^{\hbar}=\{p:A^{-1}p\cdot p\leq(\hbar/R)^{2}\}
\label{dualell}%
\end{equation}
and hence%
\begin{equation}
\{x:Ax\cdot x\leq\hbar\}^{\hbar}=\{p:A^{-1}p\cdot p\leq\hbar\}~.
\label{dualellh}%
\end{equation}

We can easily picture that the polar set $X^{\hbar}$ is \textquotedblleft
large\textquotedblright\ when $X$ is \textquotedblleft small\textquotedblright%
\ since $X$ and $X^{\hbar}$ are \textquotedblleft inversely\textquotedblright%
\ related \cite{Vershynin}; these sets can also be viewed as Fourier
transforms of each other. These qualitative statements, reminiscent of the
uncertainty principle, are clarified by the following remarkable property of
polar duality, called the \textit{Blaschke--Santal\'{o} inequality}: assume
that $X$ is a symmetric body; then there exists \cite{Campi} $c>0$ such that
\begin{equation}
c\leq\operatorname*{Vol}\nolimits_{n}(X)\operatorname*{Vol}\nolimits_{n}%
(X^{\hbar})\leq(\operatorname*{Vol}\nolimits_{n}(B^{n}(\sqrt{\hbar}))^{2}
\label{santalo1}%
\end{equation}
where $\operatorname*{Vol}\nolimits_{n}$ is the standard Lebesgue measure on
$\mathbb{R}^{n}$, and equality is attained if and only if $X\subset
\mathbb{R}_{x}^{n}$ is an ellipsoid centered at the origin The Mahler
conjecture (which is still unproven) is that the best constant $c$ is
$(4\hbar)^{n}/n!$ (see \cite{gopolar}) for a discussion of partial results and references).

\subsection{Lagrangian polar duality}

Let now $(\ell,\ell^{\prime})$ be a Lagrangian frame in the symplectic phase
space $(\mathbb{R}^{2n},\omega)$ and $X_{\ell}$ a centrally symmetric convex
body in $\ell$ (\textit{i.e.} $X_{\ell}=-X_{\ell}$). The Lagrangian polar dual
$X_{\ell^{\prime}}^{\hbar}$ of $X_{\ell}$ in $\ell^{\prime}$ is the subset of
$\ell^{\prime}$ consisting of all $z^{\prime}\in\ell^{\prime}$ such that
\begin{equation}
\omega(z^{\prime},z)\leq\hbar\text{ \ for all \ }z\in X_{\ell}; \label{ozz}%
\end{equation}
equivalently, since $X_{\ell}$ is centrally symmetric and $\omega$
antisymmetric,%
\begin{equation}
\omega(z,z^{\prime})\leq\hbar\text{ \ for all \ }z\in X_{\ell}. \label{zzbis}%
\end{equation}

The Lagrangian polar dual $X_{\ell^{\prime}}^{\hbar}$ is also a centrally
symmetric body. Suppose in particular that the Lagrangian frame $(\ell
,\ell^{\prime})$ is the canonical frame $(\ell_{X},\ell_{P})$. Then $z=(x,0)$
and $z^{\prime}=(0,p^{\prime})$ so that condition (\ref{ozz}) becomes
$p^{\prime}\cdot x\leq\hbar$; the notion of Lagrangian polar duality for
$(\ell_{X},\ell_{P})$ thus reduces the usual notion of polar duality as
described above. It is always possible to reduce Lagrangian polar duality to
ordinary polar duality. Recall that the symplectic group acts transitively on
the manifold of Lagrangian frames.

\begin{proposition}
\label{Prop3} Let $(X_{\ell},X_{\ell^{\prime}}^{\hbar})$ be a dual pair and
choose $S\in\operatorname*{Sp}(n)$ such that $(\ell,\ell^{\prime})=S(\ell
_{X},\ell_{P})$. Let $X=S^{-1}(X_{\ell})\subset\ell_{X}$. We have
$S^{-1}X_{\ell^{\prime}}^{\hbar}=X^{\hbar}\subset\ell_{P}$. Thus%
\begin{equation}
(X_{\ell},X_{\ell^{\prime}}^{\hbar})=S(X,X^{\hbar})\text{ \ \textit{if}
\ }(\ell,\ell^{\prime})=S(\ell_{X},\ell_{P}) \label{sxl}%
\end{equation}
($X^{\hbar}\subset\ell_{P}$ is the ordinary polar dual of $X\subset\ell_{X}$).
\end{proposition}

\begin{proof}
Let $z\in X_{\ell}$ and $z^{\prime}\in X_{\ell^{\prime}}^{\hbar}$ and define
$(x,0)=S^{-1}z$, $(0,p^{\prime})=S^{-1}z^{\prime}$. We have
\[
p^{\prime}\cdot x=\omega((x,0);(0,p^{\prime}))=\omega((S^{-1}z;S^{-1}%
z^{\prime})=\omega((z;z^{\prime})
\]
hence the conditions $\omega(z,z^{\prime})\leq\hbar$ \ and $p^{\prime}\cdot
x\leq\hbar$ are equivalent.
\end{proof}

The following table summarizes the main properties of Lagrangian polar duality:

\begin{center}%
\begin{tabular}
[c]{|l|l|l|}\hline
\textit{Reflexivity}: & $(X_{\ell^{\prime}}^{\hbar})_{\ell}^{\hbar}=X_{\ell}$
& LP1\\\hline
\textit{Antimonotonicity: } & $X_{\ell}\subset Y_{\ell}\Longrightarrow
Y_{\ell^{\prime}}^{\hbar}\subset X_{\ell^{\prime}}^{\hbar}$ & LP2\\\hline
\textit{Symplectic covariance}: & $S\in\operatorname*{Sp}(n)\Longrightarrow
S(X_{\ell^{\prime}}^{\hbar})=(SX_{\ell})_{S\ell^{\prime}}^{\hbar}.$ &
LP3\\\hline
\end{tabular}

\end{center}

The following characteristic property of quantum blobs is also useful:

\begin{proposition}
Let $Q=S(B^{2n}(\sqrt{\hbar}))$ be a centered quantum blob and $(\ell_{X}%
,\ell_{P})\in\operatorname*{Lag}_{0}^{2}(n)$ the canonical Lagrangian frame.
The intersection $Q\cap\ell_{X}$ and the orthogonal projection $\Pi_{\ell_{P}%
}Q$ are polar dual of each other. We have a similar statement interchanging
$\ell_{X}$ and $\ell_{P}$.
\end{proposition}

\begin{proof}
We have to show that $Q\cap\ell_{X}$ and $\Pi_{\ell_{P}}Q$ are $n$-dimensional
ellipsoids $\{x:Ax\cdot x\leq\hbar\}$ and $\{p:Bp\cdot p\leq\hbar\}$ such that
$AB=I_{n\times n}$. The quantum blob $Q$ is represented by the inequality
$Gz\cdot z\leq\hbar$ where $G=(SS^{T})^{-1}\in\operatorname*{Sp}(n)$. Writing
$G$ in block matrix form $%
\begin{pmatrix}
G_{XX} & G_{XP}\\
G_{PX} & G_{PP}%
\end{pmatrix}
$ the following relations hold in view of the symplectic conditions
(\ref{cond1}), taking into account the symmetry of $G$:%
\begin{equation}
G_{XX}G_{PX}\text{ , }G_{PX}G_{PP}\ \text{\textit{symmetric and} }G_{XX}%
G_{PP}-G_{XP}^{2}=I_{n\times n}. \label{Gcond}%
\end{equation}
With this notation we clearly have
\[
Q\cap\ell_{X}=\{x:G_{XX}x\cdot x\leq\hbar\}
\]
while the orthogonal projection $\Pi_{\ell_{P}}Q$ is given by (see
\cite{gopolar})%
\[
\Pi_{\ell_{P}}Q=\{p:(G/G_{XX})p\cdot p\leq\hbar\}
\]
where $G/G_{XX}$ is the Schur complement
\[
G/G_{XX}=G_{PP}-G_{PX}G_{XX}^{-1}G_{XP}.
\]
To prove the proposition it therefore suffices to show that%
\[
G_{XX}(G_{PP}-G_{PX}G_{XX}^{-1}G_{XP})=I_{n\times n}%
\]
but this follows from the relations (\ref{Gcond}) which in particular imply
that $G_{PX}G_{XX}^{-1}=G_{XX}^{-1}G_{PX}$:%
\begin{align*}
G_{XX}(G_{PP}-G_{PX}G_{XX}^{-1}G_{XP})  &  =G_{XX}G_{PP}-G_{XX}(G_{PX}%
G_{XX}^{-1})G_{XP})\\
&  =G_{XX}G_{PP}-G_{XP}^{2})=I_{n\times n}.
\end{align*}

\end{proof}

\section{Lagrangian Quantum States}

\subsection{Definition of a Lagrangian quantum state}

The definition of quantum states we are giving here generalizes the Definition
3 in \cite{gopolar}.

\begin{definition}
[Centered case]\label{Def1}Let $(\ell,\ell^{\prime})\in\operatorname*{Lag}%
_{0}^{2}(n)$ be a Lagrangian frame and $X_{\ell}$ be an ellipsoid with center
$0$ carried by $\ell$. We call the product $X_{\ell}\times X_{\ell^{\prime}%
}^{\hbar}$ the Lagrangian quantum state in $\mathbb{R}^{2n}$ associated with
the frame $(\ell,\ell^{\prime})$ and the ellipsoid $X_{\ell}$ and we set%
\[
\operatorname*{Quant}\nolimits_{0}(n)=\{X_{\ell}\times X_{\ell^{\prime}%
}^{\hbar}:(\ell,\ell^{\prime})\in\operatorname*{Lag}\nolimits_{0}^{2}(n)\}.
\]

\end{definition}

The elements of $\operatorname*{Quant}_{0}(1)$ are parallelograms with area
$4\hbar$ in the phase plane, while $\operatorname*{Quant}_{0}(2)$ consist of
products of two dual plane ellipses. The simplest example of a state in
$2n$-dimensional phase space is what we call the \textquotedblleft fiducial
state\textquotedblright, defined by%
\begin{equation}
X_{\ell_{X}}\times X_{\ell_{P}}^{\hbar}=B_{X}^{n}(\sqrt{\hbar})\times
B_{P}^{n}(\sqrt{\hbar}). \label{fidu}%
\end{equation}

To define a quantum state when the ellipsoid $X_{\ell}$ has center $z_{0}%
\neq0$ some care is needed. Consider for example, for $\hbar=1$, the polar
dual $X^{1}$ of the disk $X=B^{2}((a,0),1)$ in the $x,y$ plane, where $0\leq
a<1$. It is the ellipse defined by \cite{ABMB}%
\begin{equation}
(1-a^{2})^{2}\left(  p_{x}+\frac{a}{1-a^{2}}\right)  ^{2}+(1-a^{2})p_{y}%
^{2}\leq1 \label{example}%
\end{equation}
and its area $\pi/(1-a^{2})$ becomes arbitrarily large when $a$ gets close to
one. To avoid this unwanted lack of stability we proceed as follows: suppose
the ellipsoid $X_{\ell}(z_{0})$ is centered at some $z_{0}\in\ell$ and
consider the translate $X_{\ell}=T(-z_{0})X_{\ell}(z_{0})$ (it is the set of
all $z-z_{0}$ for $z\in X_{\ell}(z_{0})$). Since $X_{\ell}$ has center $0$ we
can define as usual its Lagrangian polar $X_{\ell^{\prime}}^{\hbar}$, and by
definition this will be the Lagrangian polar dual of $X_{\ell}(z_{0})$. This
procedure, has been generalized by Santal\'{o} \cite{Santalo} to arbitrary
convex bodies, but is much more complicated in this case. This leads to the
following extension of Definition \ref{Def1}:

\begin{definition}
[General case]\label{Defquant}Let $(\ell,\ell^{\prime})\in\operatorname*{Lag}%
_{0}^{2}(n)$ and $(z_{0},z_{0}^{\prime})\in\mathbb{\ell\times\ell}^{\prime}$
(cf. Lemma \ref{Rem1}). Let $X_{\ell}(z_{0})=T(z_{0})X_{\ell}$ be an ellipsoid
carried by $\ell$ and centered at $z_{0}$.The Lagrangian quantum state
associated with $(\ell,\ell^{\prime},z_{0},z_{0}^{\prime})$and $X_{\ell}$ is
the product
\begin{equation}
X_{\ell}(z_{0})\times(X_{\ell}(z_{0})-z_{0})_{\ell^{\prime}}^{\hbar}%
+z_{0}^{\prime})=X_{\ell}(z_{0})\times X_{\ell^{\prime}}^{\hbar}(z_{0}%
^{\prime}) \label{defgeneral}%
\end{equation}
where we write $X_{\ell^{\prime}}^{\hbar}(z_{0}^{\prime})=T(z_{0}^{\prime
})X_{\ell^{\prime}}^{\hbar}$. We denote $\operatorname*{Quant}(n)$ the set of
all such quantum states.
\end{definition}

Here is a basic example:

\begin{example}
Let $z_{0}=(x_{0},0)$, $z_{0}^{\prime}=(0,p_{0})$,$\ell=\ell_{X}$,
$\ell^{\prime}=\ell_{P}$, and
\[
X_{\ell}(z_{0})=T(x_{0},0)(B_{X}^{n}(\sqrt{\hbar})\times0)=B_{X}^{n}%
(x_{0},\sqrt{\hbar})\times0.
\]
We have $(B_{X}^{n}(\sqrt{\hbar})\times0)_{\ell_{P}}^{\hbar}=0\times B_{P}%
^{n}(\sqrt{\hbar})$ hence the state is
\[
(B_{X}^{n}(x_{0},\sqrt{\hbar})\times0)\times(0\times B_{P}^{n}(p_{0}%
,\sqrt{\hbar}))\equiv B_{X}^{n}(x_{0},\sqrt{\hbar})\times B_{P}^{n}%
(p_{0},\sqrt{\hbar}).
\]

\end{example}

In classical mechanics the phase space $\mathbb{R}_{x}^{n}\times\mathbb{R}%
_{p}^{n}$ can be viewed as a fiber bundle over the configuration space
$\mathbb{R}_{x}^{n}$ using the projection $\pi_{x}(x,p)=x$; the fiber is then
just the momentum space $\mathbb{R}_{x}^{n}$. In the case of Lagrangian
quantum states we have a similar situation replacing the points in
configuration space with ellipsoids (\textquotedblleft
pointillisme\textquotedblright). Let $\mathcal{E\ell\ell(}\mathbb{R}_{x}^{n})$
be the set of all ellipsoids in $\ell_{X}=\mathbb{R}_{x}^{n}$; a typical
element of $\mathcal{E\ell\ell(}\mathbb{R}_{x}^{n})$ is the set of all $x$
such that $A(x-x_{0})\cdot(x-x_{0})\leq\hbar$. For instance, $B_{X}^{n}%
(x_{0},\sqrt{\hbar})\in\mathcal{E\ell\ell(}\mathbb{R}_{x}^{n})$. Let us now
work using the canonical Lagrangian frame $(\ell_{X},\ell_{P})$ and denote by
$\operatorname*{Quant}\nolimits_{\mathrm{can}}(n)\subset\operatorname*{Quant}%
(n)$ the set of quantum states $X(x_{0},0)\times X^{\hbar}(0,p_{0})$ where
$X(x_{0},0)\subset\ell_{X}$ and $X^{\hbar}(0,p_{0})\subset\ell_{P}$ is in
$\mathcal{E\ell\ell(}\mathbb{R}_{p}^{n})$. We define a projection
$\pi_{\mathrm{can}}:\operatorname*{Quant}\nolimits_{\mathrm{can}%
}(n)\longrightarrow\mathcal{E\ell\ell(}\mathbb{R}_{x}^{n})$ by
\[
\pi_{\mathrm{can}}(X(x_{0},0)\times X^{\hbar}(0,p_{0}))=X(x_{0},0)
\]
which defines a vector bundle structure on $\operatorname*{Quant}%
\nolimits_{\mathrm{can}}(n)$. The fiber over $X(x_{0},0)\in\mathcal{E\ell
\ell(}\mathbb{R}_{x}^{n})$ is
\[
\pi_{\mathrm{can}}^{-1}(X(x_{0},0))=\{X(x_{0},0)\times X^{\hbar}%
(0,p_{0}):p_{0}\in\mathbb{R}_{p}^{n}\}
\]
so we have the identification
\[
\pi_{\mathrm{can}}^{-1}(X(x_{0},0))\equiv X(x_{0},0)\times\mathcal{E\ell\ell
(}\mathbb{R}_{p}^{n}).
\]

\subsection{Symplectic actions on $\operatorname*{Quant}_{0}(n)$}

As expected, elliptic quantum states behave well under linear or affine
symplectic transformations. Recall from Proposition \ref{Prop3} that for every
dual pair $(X_{\ell},X_{\ell^{\prime}}^{\hbar})$ there exists $S\in
\operatorname*{Sp}(n)$ such that $(\ell,\ell^{\prime})=S(\ell_{X},\ell_{P})$
and $(X_{\ell},X_{\ell^{\prime}}^{\hbar})=S(X,X^{\hbar})$. Every quantum state
$X_{\ell}\times X_{\ell^{\prime}}^{\hbar}$ is thus the image by some
$S\in\operatorname*{Sp}(n)$ of a quantum state $X\times X^{\hbar}\subset
\ell_{X}\times\ell_{P}$ associated with the canonical Lagrangian frame. The
action of $\operatorname*{Sp}(n)$ on \ $\operatorname*{Quant}_{0}(n)$ is thus
naturally defined by the formula
\begin{equation}
S^{\prime}(X_{\ell}\times X_{\ell^{\prime}}^{\hbar})=S^{\prime}S(X\times
X^{\hbar})\subset S^{\prime}S\ell_{X}\times S^{\prime}S\ell_{P}. \label{sxlh}%
\end{equation}
We have a similar definition for the action of $\operatorname*{Sp}(n)$ on
$\operatorname*{Quant}(n)$. We define the action of $S^{\prime}\in
\operatorname*{Sp}(n)$ on the state $X_{\ell}(z_{0})\times X_{\ell^{\prime}%
}^{\hbar}(z_{0}^{\prime})$ by
\begin{equation}
S^{\prime}(X_{\ell}(z_{0})\times X_{\ell^{\prime}}^{\hbar}(z_{0}^{\prime
}))=T(S^{\prime}z_{0})SX_{\ell}\times T(S^{\prime}z_{0}^{\prime}%
)(SX)_{\ell^{\prime}}^{\hbar}. \label{defsymp}%
\end{equation}
This can be rewritten, taking (\ref{sxlh}) into account,
\begin{equation}
S^{\prime}(X_{\ell}(z_{0})\times X_{\ell^{\prime}}^{\hbar}(z_{0}^{\prime
}))=(S^{\prime}SX)(S^{\prime}z_{0})\times(S^{\prime}SX^{\hbar})(S^{\prime
}z_{0}^{\prime}). \label{defsp}%
\end{equation}

\begin{proposition}
\label{Prop4}(i) The symplectic action
\begin{equation}
\operatorname*{Sp}(n)\times\operatorname*{Quant}\nolimits_{0}%
(n)\longrightarrow\operatorname*{Quant}\nolimits_{0}(n) \label{symquant1}%
\end{equation}
defined by (\ref{sxlh}) is transitive. In particular, for every state
$X_{\ell}\times X_{\ell^{\prime}}^{\hbar}$ there exists $S\in
\operatorname*{Sp}(n)$ such that%
\begin{equation}
X_{\ell}\times X_{\ell^{\prime}}^{\hbar}=S(B_{X}^{n}(\sqrt{\hbar})\times
B_{P}^{n}(\sqrt{\hbar})) \label{symquantbis}%
\end{equation}
(that $S$ is not unique: see Remark \ref{Rem2}). (ii) The symplectic action
\begin{equation}
\operatorname*{Sp}(n)\times\operatorname*{Quant}(n)\longrightarrow
\operatorname*{Quant}(n) \label{symquant2}%
\end{equation}
defined by (\ref{defsymp}) is also transitive, and there exists $S\in
\operatorname*{Sp}(n)$ such that%
\begin{equation}
X_{\ell}(z_{0})\times X_{\ell^{\prime}}^{\hbar}(z_{0}^{\prime})=S(B_{X}%
^{n}(x_{0},\sqrt{\hbar})\times B_{P}^{n}(p_{0},\sqrt{\hbar})) \label{sball}%
\end{equation}
where $x_{0}$ and $p_{0}$ are defined by: $(x_{0},0)=S^{-1}z_{0}$ and
$(0,p_{0})=S^{-1}z_{0}^{\prime}$.
\end{proposition}

\begin{proof}
To prove part \textit{(i)} it is sufficient to show that there exists
$S\in\operatorname*{Sp}(n)$ such that (\ref{symquantbis}) holds. Let now
$S\in\operatorname*{Sp}(n)$ be such that $(\ell,\ell^{\prime})=S(\ell_{X}%
,\ell_{P})$ and $(X_{\ell},X_{\ell^{\prime}}^{\hbar})=S(X,X^{\hbar})$. There
exists a symmetric positive definite matrix $A$ such that ellipsoid $X$ is
$A^{-1/2}(B_{X}^{n}(\sqrt{\hbar}))$ hence $X^{\hbar}=A^{1/2}(B_{X}^{n}%
(\sqrt{\hbar}))$ and
\[
X\times X^{\hbar}=M_{A^{1/2}}(B_{X}^{n}(\sqrt{\hbar})\times B_{P}^{n}%
(\sqrt{\hbar}))
\]
where $M_{A^{1/2}}=%
\begin{pmatrix}
A^{1/2} & 0\\
0 & A^{-1/2}%
\end{pmatrix}
\in\operatorname*{Sp}(n)$ so that we have
\[
(X_{\ell},X_{\ell^{\prime}}^{\hbar})=SM_{A^{1/2}}(B_{X}^{n}(\sqrt{\hbar
})\times B_{P}^{n}(\sqrt{\hbar}))
\]
which was to be proven. Part \textit{(ii)} is proven in a similar way.
\end{proof}

\subsection{$\operatorname*{Quant}_{0}(n)$ as a homogeneous space}

Proposition \ref{Prop4} leads a topological identification of
$\operatorname*{Quant}_{0}(n)$ with the homogeneous space $\operatorname*{Sp}%
(n)/O(n)$. We begin by noting that the \textquotedblleft fiducial quantum
state\textquotedblright\ $B_{X}^{n}(\sqrt{\hbar})\times B_{P}^{n}(\sqrt{\hbar
})$ is invariant\ by the action of the subgroup $O(n)$ of $U(n)$ consisting of
all matrices $M_{H}=%
\begin{pmatrix}
H & 0\\
0 & H
\end{pmatrix}
$ with $H\in O(n,\mathbb{R})$.

\begin{remark}
The quotient $\operatorname*{Sp}(n)/U(n)$ (which is \textquotedblleft
smaller\textquotedblright\ than $\operatorname*{Sp}(n)/O(n)$) can be
identified with the set of Wigner transforms of Gaussian wavepackets
(\cite{Littlejohn}, formula (8.12)). This shows that $\operatorname*{Quant}%
\nolimits_{0}(n)$ contains more information than the Gaussian wavepackets
which we will study below.
\end{remark}

Let us state things in a more precise way. We first note that the
\textquotedblleft orthogonal symplectic group\textquotedblright\ $O(n)$ is the
largest subgroup of $\operatorname*{Sp}(n)$ such that%
\[
S(B_{X}^{n}(\sqrt{\hbar})\times B_{P}^{n}(\sqrt{\hbar}))=B_{X}^{n}(\sqrt
{\hbar})\times B_{P}^{n}(\sqrt{\hbar}),
\]
\textit{i.e.} $O(n)$ is the stabilizer (or isotropy subgroup) of the action of
$\operatorname*{Sp}(n)$ on $B_{X}^{n}(\sqrt{\hbar})\times B_{P}^{n}%
(\sqrt{\hbar})$ (we are identifying, as usual, $B_{X}^{n}(\sqrt{\hbar}%
)\times0\subset\ell_{X}$ with $B_{X}^{n}(\sqrt{\hbar})$ and $0\times B_{X}%
^{n}(\sqrt{\hbar})\times\subset\ell_{P}$ with $B_{P}^{n}(\sqrt{\hbar})$. To
see this it suffices to note that if $S(B_{X}^{n}(\sqrt{\hbar}))=B_{X}%
^{n}(\sqrt{\hbar})$ and similarly $S(B_{P}^{n}(\sqrt{\hbar}))=B_{P}^{n}%
(\sqrt{\hbar})$ then we must have, by homogeneity, $S\ell_{X}=\ell_{X}$ and
$S\ell_{P}=\ell_{P}$, hence we must have $S=%
\begin{pmatrix}
H & 0\\
0 & H
\end{pmatrix}
$ for some $H\in O(n)$. Since $\operatorname*{Sp}(n)$ is a classical Lie group
and $O(n)$ is a closed subgroup it follows from the theory of homogeneous
spaces that we have the identification%
\begin{equation}
\operatorname*{Quant}\nolimits_{0}(n)\equiv\operatorname*{Sp}%
(n)/O(n)\label{quanton}%
\end{equation}
which allows to define a topology on $\operatorname*{Quant}\nolimits_{0}(n)$
and hence a fiber bundle \cite{Steenrod}
\[
\mathcal{F}=(\operatorname*{Sp}(n),\operatorname*{Quant}\nolimits_{0}%
(n),\pi_{0}^{\operatorname*{Quant}},O(n))
\]
with projection%
\[
\pi_{0}^{\operatorname*{Quant}}:\operatorname*{Sp}(n)\longrightarrow
\operatorname*{Quant}\nolimits_{0}(n).
\]

\begin{remark}
The complex structure rotation $J:(x,p)\longmapsto(p,-x)$ also fixes the
Lagrangian product of two same size balls, but does not belong to the group
$O(n)$. On the analytical level $J$ plays the role of a Fourier transform.
\end{remark}

\section{$\operatorname*{Quant}(n)$ and Gaussian Wavepackets}

In this section we identify a subset of $\operatorname*{Quant}(n)$ with the
set of all Gaussian wavepackets.

\subsection{John and L\"{o}wner ellipsoids}

There is a vast literature on the L\"{o}wner and John ellipsoids of a convex
body; a classical reference is \cite{Ball}. Let $X$ be a convex body in any
Euclidean space $\mathbb{R}^{n}$. The L\"{o}wner ellipsoid $X_{\mathrm{L\ddot
{o}wner}}$ of $X$ \textit{is the unique ellipsoid in} $\mathbb{R}^{n}$
\textit{with minimum volume containing} $X$ and the John ellipsoid
$X_{\mathrm{John}}$ \textit{is the unique ellipsoid in} $\mathbb{R}^{n}$
\textit{with maximum volume contained in} $X$. If $A$ is an invertible linear
mapping then
\begin{equation}
(A(X))_{\mathrm{L\ddot{o}wner}}=A(X_{\mathrm{L\ddot{o}wner}})\text{ ,
}(A(X))_{\mathrm{John}}=A(X_{\mathrm{John}}) \label{JL1}%
\end{equation}

Not so surprisingly, if $X$ is a centrally symmetric convex body, then
$X_{\mathrm{John}}$ and $X_{\mathrm{L\ddot{o}wner}}$ are polar duals of each
other in the following sense \cite{ABMB}:%
\begin{equation}
(X_{\mathrm{John}})^{\hbar}=(X^{\hbar})_{\mathrm{L\ddot{o}wner}}\text{ \ ,
\ }(X_{\mathrm{L\ddot{o}wner}})^{\hbar}=(X^{\hbar})_{\mathrm{John}%
}.\label{JL2}%
\end{equation}
This property extends to Lagrangian polar duality. Let $(\ell,\ell^{\prime})$
be a Lagrangian frame and $(X_{\ell},X_{\ell^{\prime}}^{\hbar})$ a dual pair
of centered convex bodies. Then
\begin{equation}
((X_{\ell})_{\mathrm{John}})_{\ell^{\prime}}^{\hbar}=(X_{\ell^{\prime}}%
^{\hbar})_{\mathrm{L\ddot{o}wner}}\text{ \ , \ }((X_{\ell})_{\mathrm{L\ddot
{o}wner}})_{\ell^{\prime}}^{\hbar}=(X_{\ell^{\prime}}^{\hbar})_{\mathrm{John}%
}.\label{JL3}%
\end{equation}

The following particular case will be very important for what follows. We
denote $B_{X}^{n}(R)$ (resp. $B_{P}^{n}(R)$) the ball $|x|\leq R$ (resp.
$|p|\leq R$) in position (resp. momentum) space.

\begin{proposition}
\label{PropJohn}The John ellipsoid of $B_{X}^{n}(R)\times B_{P}^{n}(R)$ is
$B^{2n}(R)$. In particular%
\begin{equation}
\left(  B_{X}^{n}(\sqrt{\hbar})\times B_{P}^{n}(\sqrt{\hbar})\right)
_{\mathrm{John}}=B^{2n}(\sqrt{\hbar}). \label{BXP}%
\end{equation}

\end{proposition}

\begin{proof}
The inclusion
\begin{equation}
B^{2n}(R)\subset B_{X}^{n}(R)\times B_{P}^{n}(R) \label{incl}%
\end{equation}
is obvious, and we cannot have
\[
B^{2n}(R^{\prime})\subset B_{X}^{n}(R)\times B_{P}^{n}(R)
\]
if $R^{\prime}>R$. Assume now that the John ellipsoid $\Omega_{\mathrm{John}}$
of $\Omega=B_{X}^{n}(R)\times B_{P}^{n}(R)$ is defined by $Ax^{2}%
+Bxp+Cp^{2}\leq R^{2}$ where $A,C>0$ and $B$ are real $n\times n$ matrices.
Since $\Omega$ is invariant by the transformation $(x,p)\longmapsto(p,x)$ so
is $\Omega_{\mathrm{John}}$ and we must thus have $A=C$ and $B=B^{T}$.
Similarly, $\Omega$ being invariant by the partial reflection
$(x,p)\longmapsto(-x,p)$ we get $B=0$ so $\Omega_{\mathrm{John}}$ is defined
by $Ax^{2}+Ap^{2}\leq R^{2}$. We next observe that $\Omega$ and hence
$\Omega_{\mathrm{John}}$ are invariant under the symplectic transformations
$(x,p)\longmapsto(Hx,HP)$ where $H\in O(n,\mathbb{R})$ so we must have $AH=HA$
for all $H\in O(n,\mathbb{R})$, but this is only possible if $A=\lambda
I_{n\times n}$ for some $\lambda\in\mathbb{R}$. The John ellipsoid is thus of
the type $B^{2n}(R/\sqrt{\lambda})$ for some $\lambda\geq1$ and this concludes
the proof in view of (\ref{incl}) since the case $\lambda>R^{2}$ is excluded.
\end{proof}

\subsection{Gaussian wavepackets and their Wigner transforms}

Recall \cite{Wigner} that the Wigner transform of a square integrable function
$\psi$ on $\mathbb{R}_{x}^{n}$ is defined by the absolutely convergent
integral
\begin{equation}
W\psi(x,p)=\left(  \tfrac{1}{2\pi\hbar}\right)  ^{n}\int e^{-\frac{i}{\hbar
}py}\psi(x+\tfrac{1}{2}y)\psi^{\ast}(x-\tfrac{1}{2}y)d^{n}y. \label{wigtra}%
\end{equation}
The Wigner transform is a real function which can take negative values (except
when $\psi$ is a Gaussian). We recall the \textquotedblleft marginal
properties\textquotedblright\ of the Wigner transform: if both $\psi$ and its
Fourier transform
\[
\widehat{\psi}(p)=F\psi(p)=\left(  \tfrac{1}{2\pi\hbar}\right)  ^{n/2}\int
e^{-\frac{1}{\hbar}p\cdot x}\psi(x)d^{n}x
\]
are in $L^{1}(\mathbb{R}_{x}^{n})\cap L^{2}(\mathbb{R}_{x}^{n})$ then
\begin{align}
\int W\psi(x,p)d^{n}p  &  =|\psi(x)|^{2}\label{marg1}\\
\int W\psi(x,p)d^{n}x  &  =|F\psi(p)|^{2}. \label{marg2}%
\end{align}
These relations imply that
\begin{equation}
\int W\psi(x,p)d^{n}pd^{n}x=||\psi||_{L^{2}} \label{norm}%
\end{equation}
so that if $\psi$ is normalized to one then the integral of $W\psi$ over all
of phase space is equal to one. These properties motivate the interpretation
of the Wigner transform as a quasi-probability density.

A crucial fact is that the Wigner transform enjoys the property of symplectic
covariance \cite{Birk,Wigner}, that is, we have for every $S\in
\operatorname*{Sp}(n)$,
\begin{equation}
W\psi(S^{-1}z)=W(\widehat{S}\psi)(z) \label{sympco}%
\end{equation}
where $\widehat{S}$ is anyone of the two metaplectic operators covering $S$.
This property is instrumental in proving the symplectic covariance of Weyl
quantization, and implies that the metaplectic group acts transitively on the
Gaussian wavepackets we define below.

Following our work in \cite{go09} we introduced in \cite{blob} the notion of
\textquotedblleft quantum blob\textquotedblright. Their properties were
detailed in our \textit{Phys. Reps.} paper \cite{goluPR} with F. Luef. A
quantum blob is the image of a phase space ball $B^{2n}(z_{0},\sqrt{\hbar
}):|z-z_{0}|\leq\sqrt{\hbar}$ by some $S\in\operatorname*{Sp}(n)$. it can be
viewed as the smallest phase space unit compatible with the uncertainty
principle expressed in terms of variances and covariances (for a discussion of
the relevance of the use of standard deviations to formulate the uncertainty
relations see \cite{hiuf}). It turns out that there is a canonical
correspondence between quantum blobs and Gaussian wavepackets%
\begin{equation}
\psi_{AB}(x)=e^{i\gamma}\left(  \tfrac{1}{\pi\hbar}\right)  ^{n/4}(\det
A)^{1/4}e^{-\tfrac{1}{2\hbar}(A+iB)x\cdot x} \label{wpt}%
\end{equation}
and their displacements $\psi_{AB,z_{0}}=\widehat{T}(z_{0})\psi_{AB}$ by the
Heisenberg--Weyl operator $\widehat{T}(z_{0})$ \cite{Birk,Littlejohn}. In
(\ref{wpt}) $A$ and $B$ are real symmetric $n\times n$ matrices with $A$
positive definite and $\gamma\in R$ an arbitrary constant phase; we will not
care about the value of this phase factor since we will be dealing with the
properties of the quantum states $|\psi_{AB}\rangle$. When $A=I$ (the
identity), $B=0$, and $\gamma=0$ the function $\psi_{AB}$ reduces to the
\textquotedblleft fiducial coherent state\textquotedblright\ (we are using the
terminology in \cite{Littlejohn}):
\begin{equation}
\phi_{0}(x)=(\pi\hbar)^{-n/4}e^{-|x|^{2}/2\hbar}. \label{fid}%
\end{equation}
It turns out that every Gaussian wavepacket (\ref{wpt}) can be obtained from
the fiducial state by using metaplectic operators.

We will denote by $\operatorname*{Gauss}(n)$ the set of all Gaussian
wavepackets $\widehat{T}(z_{0})\psi_{AB}$, and by $\operatorname*{Gauss}%
_{0}(n)$ the subset consisting of centered wavepackets. One shows
\cite{Bas,Birk}, using the symplectic covariance formula (\ref{sympco}), that
the Wigner transform of $\widehat{T}(z_{0})\psi_{AB}$ is the phase space
Gaussian%
\begin{equation}
W\psi_{AB}(z)=(\pi\hbar)^{-n}e^{-\tfrac{1}{\hbar}G(z-z_{0})\cdot(z-z_{0}%
)}\label{phagauss}%
\end{equation}
where $G$ is the positive definite symmetric and symplectic $2n\times2n$
matrix
\begin{equation}
G=(S_{AB}S_{AB}^{T})^{-1}\text{ \ , \ }S_{AB}=%
\begin{pmatrix}
A^{-1/2} & 0\\
-BA^{-1/2} & A^{1/2}%
\end{pmatrix}
.\label{gaga}%
\end{equation}

Let us denote by $\operatorname*{QB}(n)$ the set of all quantum blobs
$S(B^{2n}(z_{0},\sqrt{\hbar}))$, $S\in\operatorname*{Sp}(n)$ and by
$\operatorname*{QB}_{0}(n)$ the subset consisting of all centered quantum
blobs $S(B^{2n}(\sqrt{\hbar}))$. Recall that $\operatorname*{Gauss}(n)$ is the
set of all Gaussian states $\widehat{T}(z_{0})\psi_{AB}$.

\begin{proposition}
(i) There is a bijective correspondence $\operatorname*{Gauss}%
(n)\longleftrightarrow\operatorname*{QB}(n)$; it is defined by%
\[
\widehat{T}(z_{0})\psi_{AB}\longrightarrow T(z_{0})S_{AB}(B^{2n}(\sqrt{\hbar
})).
\]
where $T(z_{0})$ is the phase space translation $z\longmapsto z+z_{0}$ and
$S_{AB}\in\operatorname*{Sp}(n)$ is defined by (\ref{phagauss}) and
(\ref{gaga}). (ii) The transitive action of $\operatorname*{Sp}(n)$ on the set
$\operatorname*{QB}_{0}(n)$ of centered quantum blobs induces a transitive
action of $\operatorname*{Mp}(n)$ on $\operatorname*{Gauss}_{0}(n).$ More
generally the transitive action of the inhomogeneous symplectic group
$\operatorname{ISp}(n)$ on $\operatorname*{QB}(n)$ induces a transitive action
of $\operatorname*{IMp}(n)$ on $\operatorname*{Gauss}(n)$.
\end{proposition}

\begin{proof}
(i) In view of the discussion above the Wigner transform associates to
$\widehat{T}(z_{0})\psi_{AB}$ the phase space ellipsoid
\[
Q=\{z:G_{AB}(z-z_{0})\cdot(z-z_{0})\leq\hbar\}
\]
where $G=(S_{AB}S_{AB}^{T})^{-1}$ hence $Q$ is the quantum blob $T(z_{0}%
)S_{AB}(B^{2n}(\sqrt{\hbar}))$. Let us show that, conversely, every quantum
blob is is obtained from a unique state $|\widehat{T}(z_{0})\psi_{AB}\rangle$.
Let $Q=T(z_{0})S(B^{2n}(\sqrt{\hbar}))$be a quantum blob, that is
\[
Q=\{z:G(z-z_{0})\cdot(z-z_{0})\leq\hbar\}\text{ \ },\text{ \ }G=(SS^{T}%
)^{-1}.
\]
To $Q$ we associate the function $\psi$ with Wigner transform
\[
W\psi(z)=(\pi\hbar)^{-n}e^{-\tfrac{1}{\hbar}G(z-z_{0})\cdot(z-z_{0})}.
\]
We have
\[
W\psi(S(z+S^{-1}z_{0}))=(\pi\hbar)^{-n}e^{-\tfrac{1}{\hbar}|z|^{2}}=W\phi
_{0}(z)
\]
hence, by the symplectic covariance formula (\ref{sympco}),%
\[
W(\widehat{S}\psi)(z)=W\phi_{0}(z-S^{-1}z_{0})=W(\widehat{T}(S^{-1}z_{0}%
)\phi_{0})(z)
\]
where $\widehat{S}\in\operatorname*{Mp}(n)$ covers $S$. It follows that we
have
\[
\widehat{S}\psi(x)=e^{i\gamma}\widehat{T}(S^{-1}z_{0})\phi_{0}(x)
\]
that is%
\[
\psi(x)=e^{i\gamma}\widehat{S}\widehat{T}(S^{-1}z_{0})\phi_{0}(x)=e^{i\gamma
}\widehat{T}(z_{0})\widehat{S}\phi_{0}(x)
\]
so that $\psi=e^{i\gamma}\widehat{T}(z_{0})\psi_{A,B}$ for some (uniquely
defined) matrices $A$ and $B$. (ii). see \cite{Birk}.
\end{proof}

For a detailed study of the correspondence $\operatorname*{Gauss}%
(n)\longleftrightarrow\operatorname*{QB}(n)$ see \cite{blob,goluPR}.

\subsection{Construction of a Quantum Gaussian Space}

Consider first the very simple case where $X$ is the ball $B_{X}^{n}%
(\sqrt{\hbar})$ whose polar dual is $X^{\hbar}=B_{P}^{n}(\sqrt{\hbar})$. The
corresponding elliptic quantum state is the product $B_{X}^{n}(\sqrt{\hbar
})\times B_{P}^{n}(\sqrt{\hbar})$. In view of Proposition \ref{PropJohn} the
John ellipsoid of this state is $B^{2n}(\sqrt{\hbar})$, and to the latter
corresponds the fiducial coherent state $\phi_{0}(x)=(\pi\hbar)^{-n/4}%
e^{-|x|^{2}/2\hbar}$. Slightly more generally, let $U$ be a symplectic
rotation and define a Lagrangian frame $(\ell,\ell^{\prime})$ by $\ell
=U\ell_{X}$ and $\ell^{\prime}=U\ell_{P}$. Identifying $B_{X}^{n}(\sqrt{\hbar
})$ with $B_{X}^{n}(\sqrt{\hbar})\times0\subset\ell_{X}$ the rotation $U$
takes this set to $U(B_{X}^{n}(\sqrt{\hbar})\times0)\subset\ell$ and,
similarly, $U(B_{P}^{n}(\sqrt{\hbar})\times0)\subset\ell^{\prime}$. The state
$B_{X}^{n}(\sqrt{\hbar})\times B_{P}^{n}(\sqrt{\hbar})$ is replaced with
$U(B_{X}^{n}(\sqrt{\hbar})\times B_{P}^{n}(\sqrt{\hbar}))$ whose John
ellipsoid is, by rotational symmetry,
\[
\left(  U(B_{X}^{n}(\sqrt{\hbar})\times B_{P}^{n}(\sqrt{\hbar}))\right)
_{\mathrm{John}}=U(B^{2n}(\sqrt{\hbar}))=B^{2n}(\sqrt{\hbar})
\]
in view of the linear transformation property (\ref{JL1}). The states
$B_{X}^{n}(\sqrt{\hbar})\times B_{P}^{n}(\sqrt{\hbar}$ and $U(B_{X}^{n}%
(\sqrt{\hbar})\times B_{P}^{n}(\sqrt{\hbar}))$ thus have the \emph{same} John
ellipsoid, and to both states thus corresponds the fiducial Gaussian
wavepacket $\phi_{0}$. From the Wigner transform point of view, this property
just reflects the rotational invariance of $\phi_{0}$: we have
\[
W\phi_{0}(Uz)=(\pi\hbar)^{-n}e^{-\frac{1}{\hbar}Uz\cdot Uz}=(\pi\hbar
)^{-n}e^{-\frac{1}{\hbar}z\cdot z}=W\phi_{0}(z).
\]
Consider next the slightly more general case where $X$ is the ellipsoid%
\[
X=\{x:Ax\cdot x\leq\hbar\}=A^{-1/2}(B_{X}^{n}(\sqrt{\hbar}))
\]
with $A=A^{T}>0$; hence
\[
X^{\hbar}=\{p:A^{-1}p\cdot p\leq\hbar\}=A^{1/2}(B_{P}^{n}(\sqrt{\hbar}))
\]
and the corresponding quantum state is then
\[
A^{-1/2}(B_{X}^{n}(\sqrt{\hbar}))\times A^{1/2}(B_{P}^{n}(\sqrt{\hbar
}))=M_{A^{1/2}}(B_{X}^{n}(\sqrt{\hbar})\times B_{P}^{n}(\sqrt{\hbar}))
\]
where $M_{A^{1/2}}=%
\begin{pmatrix}
A^{1/2} & 0\\
0 & A^{-1/2}%
\end{pmatrix}
$ is a symplectic dilation. Using again (\ref{JL1}) the John ellipsoid of this
state is
\[
(X\times X^{\hbar})_{\mathrm{John}}=M_{A^{1/2}}(B^{2n}(\sqrt{\hbar}))
\]
and to the latter corresponds the function with Wigner transform
\[
W\psi(z)=(\pi\hbar)^{-n}\exp-\left[  \frac{1}{\hbar}(Ax\cdot x+A^{-1}p\cdot
p)\right]
\]
and hence, up to an irrelevant constant phase $e^{i\gamma}$,%
\[
\psi(x)=\psi_{A,0}(x)=\left(  \tfrac{1}{\pi\hbar}\right)  ^{n/4}(\det
A)^{1/4}e^{-\tfrac{1}{2\hbar}Ax\cdot x}.
\]

These examples suggest that there is a deeper underlying structure relating
elliptic quantum states to Gaussian wavefunctions. To study this relation we
begin by defining an equivalence relation on $\operatorname*{Quant}_{0}(n)$:
We will say that two states $X_{\ell_{1}}\times X_{\ell_{1}^{\prime}}^{\hbar}$
and $X_{\ell_{2}}\times X_{\ell_{2}^{\prime}}^{\hbar}$ are unitarily
equivalent and write
\[
X_{\ell_{1}}\times X_{\ell_{1}^{\prime}}^{\hbar}\overset{U(n)}{\sim}%
X_{\ell_{2}}\times X_{\ell_{2}^{\prime}}^{\hbar}%
\]
if there exists a symplectic rotation $U\in U(n)$ such that $(\ell_{1}%
,\ell_{1}^{\prime})=U(\ell_{2},\ell_{2}^{\prime})$ and
\[
X_{\ell_{1}}\times X_{\ell_{1}^{\prime}}^{\hbar}=U(X_{\ell_{2}}\times
X_{\ell_{2}^{\prime}}^{\hbar}).
\]
Since $U(n)$ is a group the relation $\overset{U(n)}{\sim}$ enjoys the
properties of reflexivity, symmetry, and transitivity, and is thus indeed an
equivalence relation. We denote by $\widetilde{X_{\ell}\times X_{\ell^{\prime
}}^{\hbar}}$ the equivalence class of the state $X_{\ell}\times X_{\ell
^{\prime}}^{\hbar}$ for this relation and by $\operatorname*{Quant}%
\nolimits_{0}(n)/U(n)$ the set of all such equivalence classes. Recall
(formula (\ref{quanton})) that we have identified $\operatorname*{Quant}%
\nolimits_{0}(n)$ with $\operatorname*{Sp}(n)/O(n)$. Following result
identifies $\operatorname*{Gauss}\nolimits_{0}(n)$ with $\operatorname*{Quant}%
\nolimits_{0}(n)/U(n)$:

\begin{proposition}
\label{Thm1}There is a canonical identification
\begin{equation}
\operatorname*{Gauss}\nolimits_{0}(n)\equiv\operatorname*{Quant}%
\nolimits_{0}(n)/U(n) \label{gausson}%
\end{equation}
between the set of centered Gaussian wavepackets $\psi_{AB}$ and the
equivalence classes $\widetilde{X_{\ell}\times X_{\ell^{\prime}}^{\hbar}}$ of
centered elliptic quantum states. More generally we have an identification%
\begin{equation}
\operatorname*{Gauss}(n)\equiv\operatorname*{Quant}(n)/U(n) \label{gaussQuant}%
\end{equation}

\end{proposition}

\begin{proof}
Let $\psi_{A,B}\in\operatorname*{Gauss}\nolimits_{0}(n)$ be a Gaussian
wavepacket and
\[
W\psi_{AB}(z)=(\pi\hbar)^{-n}e^{-\frac{1}{\hbar}Gz\cdot z}\text{ , }%
G=(SS^{T})^{-1}%
\]
its Wigner transform. The ellipsoid $\{z:Gz\cdot z\leq\hbar\}$ is the quantum
blob $Q=S(B^{2n}(\sqrt{\hbar}))$, and in view of Proposition \ref{PropJohn}
the latter is the John ellipsoid of the state%
\begin{gather*}
X_{\ell}\times X_{\ell^{\prime}}^{\hbar}=S(B_{X}^{n}(\sqrt{\hbar})\times
B_{P}^{n}(\sqrt{\hbar})),\text{ }\\
\text{\ }\ell=S\ell_{X}\text{, }\ell^{\prime}=S\ell_{P}.
\end{gather*}
If $S^{\prime}\in\operatorname*{Sp}(n)$ is another symplectic matrix such that
$G=(S^{\prime}(S^{\prime})^{T})^{-1}$ then $S^{\prime}=SU$ for some symplectic
rotation $U\in U(n)$ and hence $S^{\prime}(B^{2n}(\sqrt{\hbar}))=S(B^{2n}%
(\sqrt{\hbar}))$ so that $Q$ is also the John ellipsoid of the state%
\begin{gather*}
X_{\ell_{1}}\times X_{\ell_{1}^{\prime}}^{\hbar}=S^{\prime}(B_{X}^{n}%
(\sqrt{\hbar})\times B_{P}^{n}(\sqrt{\hbar})),\\
\text{\ }\ell_{1}=SU\ell_{X}\text{, }\ell_{1}^{\prime}=SU\ell_{P}.
\end{gather*}
Conversely, let $X_{\ell}\times X_{\ell^{\prime}}^{\hbar}$ be a centered
elliptic quantum state and choose $S\in\operatorname*{Sp}(n)$ such that
$(\ell,\ell^{\prime})=S(\ell_{X},\ell_{P})$ and%
\begin{equation}
X_{\ell}\times X_{\ell^{\prime}}^{\hbar}=S(B_{X}^{n}(\sqrt{\hbar})\times
B_{P}^{n}(\sqrt{\hbar}))\label{xlxl}%
\end{equation}
(Proposition \ref{Prop4}). In view of Proposition \ref{PropJohn} the John
ellipsoid of $X_{\ell}\times X_{\ell^{\prime}}^{\hbar}$ is the quantum blob
$Q=S(B^{2n}(\sqrt{\hbar}))$, hence to $X_{\ell}\times X_{\ell^{\prime}}%
^{\hbar}$ corresponds the generalized Gaussian $\psi_{AB}$ with Wigner
transform%
\[
W\psi_{AB}(z)=(\pi\hbar)^{-n}e^{-\frac{1}{\hbar}Gz\cdot z}\text{ , }%
G=(SS^{T})^{-1}.
\]
We may replace $X_{\ell}\times X_{\ell^{\prime}}^{\hbar}$ with%
\[
X_{\ell_{1}}\times X_{\ell_{1}^{\prime}}^{\hbar}=S^{\prime}U(B_{X}^{n}%
(\sqrt{\hbar})\times B_{P}^{n}(\sqrt{\hbar})),\text{ \ }\ell_{1}=SU\ell
_{X}\text{, }\ell_{1}^{\prime}=SU\ell_{P}.
\]
with $U\in U(n)$ without altering $G$, hence $W\psi_{AB}$ (and thus $\psi
_{AB}$) only depends on the equivalence class $\widetilde{X_{\ell}\times
X_{\ell^{\prime}}^{\hbar}}$. The extension of (\ref{gausson}) to formula
(\ref{gaussQuant}) is straightforward.
\end{proof}

\section{Perspectives for a Generalization}

Sofar we have been considering Lagrangian products of ellipsoids. The next
--and fundamental!-- step would be to generalize our constructions to products
$X\times X^{\hbar}$ (or, more generally, $X_{\ell}\times X_{\ell^{\prime}%
}^{\hbar}$) where $X$ or $X_{\ell}$ is not an ellipsoid, but an arbitrary
convex set, leading to a non-Gaussian quantum state. It is clear why tis
problem has such an overwhelming importance since it opens the door to a
general geometric theory of quantum states. This problem will be addressed in
forthcoming work: let us just outline here some of the difficulties inherent
to such an extension of our theory.  So, we would like to construct a
generalization of $\operatorname*{Quant}(n)$ where the Lagrangian quantum
states are represented by arbitrary convex sets. The first mathematical
difficulty that occurs is the determination of the point with respect to which
the polar dual should be calculated. Let in fact $X(x_{0})$ be an arbitrary
convex body in $\ell_{X}=\mathbb{R}_{x}^{n}$; by definition its centroid (or
barycenter) is
\begin{equation}
x_{0}=\frac{1}{\operatorname*{Vol}\nolimits_{n}(X)}\int_{X}x_{1}dx_{1}%
+\cdot\cdot\cdot+x_{n}dx_{n}=0.\label{centroid}%
\end{equation}
It is easily verified that if $X$ is an ellipsoid, then the centroid coincides
with its center in the usual sense. To define the polar dual of $X(x_{0})$ one
is tempted to use the same procedure as for ellipsoids and to define
$X(x_{0})^{\hbar}$ as the dual of the centered convex body $X=T(-x_{0}%
)X(x_{0})$. However this is not the good choice. Here is why: when we defined
the polar of an ellipsoid by translating it is centered at the origin it turns
out that the Blaschke--Santal\'{o} product $\operatorname*{Vol}\nolimits_{n}%
(X(x_{0}))\operatorname*{Vol}\nolimits_{n}(X^{\hbar}(x_{0}))$ attains the
value $(\operatorname*{Vol}\nolimits_{n}B^{n}(\sqrt{\hbar})^{2}$. The
difficulty comes from the fact that in the general case of arbitrary convex
body we need to choose the correct center with respect to which the polarity
is defined since there is no privileged \textquotedblleft
center\textquotedblright; different choices may lead to polar duals with very
different volumes (see Example \ref{example}). \ Santal\'{o} proved in
\cite{Santalo} the following remarkable result: there exists a \emph{unique}
interior point $x_{\mathrm{S}}$ of $X$ (the \textquotedblleft Santal\'{o}
point of $X$\textquotedblright) such that the polar dual $X^{\hbar
}(x_{\mathrm{S}})=(X-x_{\mathrm{S}})^{\hbar}$ has centroid $\overline{p}=0$
and its volume $\operatorname*{Vol}\nolimits_{n}(X^{\hbar}(x_{\mathrm{S}}))$
is \textit{minimal} for all possible interior points $x_{0}$:
\begin{equation}
\operatorname*{Vol}\nolimits_{n}(X)\operatorname*{Vol}\nolimits_{n}(X^{\hbar
}(x_{\mathrm{S}}))\leq(\operatorname*{Vol}\nolimits_{n}B^{n}(\sqrt{\hbar}%
)^{2}\label{sant2}%
\end{equation}
with equality if and only if $X$ is an ellipsoid. We note that the practical
determination of the Santal\'{o} point is in general difficult and one has to
use ad hoc methods in each particular case. See \cite{arkami} for a discussion
of this issue.

Having in mind that the polar dual is calculated with respect to the
Santal\'{o} point (not the centroid!) we can define the associated canonical
Lagrangian quantum state exactly as follows let $(\ell_{X},\ell_{P}%
)\in\operatorname*{Lag}_{0}^{2}(n)$, $be$ the canonical Lagrangian frame and
$X(x_{\mathrm{S}})\in\operatorname*{Conv}(\ell)$ a convex body carried by
$\ell_{X}$ and with Santal\'{o} point $x_{\mathrm{S}}$. The associated
Lagrangian state is then
\[
X(x_{\mathrm{S}})\times(X(x_{\mathrm{S}})-x_{\mathrm{S}})^{\hbar}%
+p_{0})=X(x_{\mathrm{S}})\times X^{\hbar}(p_{0})
\]
and we again have a fiber bundle structure%
\[
\pi:\operatorname*{Quant}\nolimits_{0}(n)\longrightarrow\operatorname*{Conv}%
(\ell).
\]
The study of the latter is less straightforward than in the case of
ellipsoids, and will be done in a forthcoming work. We also notice that we can
associated to every state an ellipsoid using the John ellipsoid method, but
the role played by the latter is unclear (it is not quite obvious that it
should be a quantum blob; if it were the case it could correspond to the
covariance matrix of the state). At this point one might want to use tie
theory of the Minkowski functional to give a geometric study; this leas us to
consider non-linear problems. All this is also related to the powerful notion
of symplectic capacity, which we discussed in \cite{BSM} following the ideas
in \cite{arkami,armios08,arkaos13} All these questions are fascinating ad
answers might lead to a geometric reformulation of quantum mechanics where the
notion of polar duality in a sense replaces the usual uncertainty principle.
We will come back with answers in a near future.

\begin{acknowledgement}
Maurice de Gosson has been financed by the Grant P 33447 N of the Austrian
Science Fund FWF.
\end{acknowledgement}

\begin{itemize}
\item The author declares that there are no conflicts of interests in this work;

\item No experimentation on human beings or animals has been performed for
this work;

\item No numerical or experimental data have been produced or used.
\end{itemize}

\end{document}